\documentclass[a4paper,12pt]{article}

\usepackage[utf8]{inputenc}
\usepackage[T1]{fontenc}
\usepackage{amsmath,amssymb,amsthm}
\usepackage{fullpage}
\usepackage{graphicx}
\usepackage{mathtools}
\usepackage{paralist}
\usepackage{nicefrac}
\usepackage[hidelinks]{hyperref}
\usepackage{enumitem}
\usepackage{tikz}
\usetikzlibrary{calc,fadings}

\theoremstyle{plain}
\newtheorem{theorem}{Theorem}
\newtheorem{proposition}{Proposition}

\newtheorem{lemma}[theorem]{Lemma}

\theoremstyle{definition}

\newcommand{\dist}{\mathrm{dist}}

\newcommand{\len}{\mathrm{length}}

\newcommand{\bd}{\partial}

\renewcommand{\leq}{\leqslant}
\renewcommand{\geq}{\geqslant}
\renewcommand{\rho}{\varrho}

\newcommand{\eps}{\varepsilon}

\newcommand{\Reals}{\mathbb{R}}

\title{Online Search for a Hyperplane\\
in High-Dimensional Euclidean Space}

\author{Antonios Antoniadis\thanks{Department of Applied Mathematics, University of Twente, Enschede, The Netherlands. \texttt{\{a.antoniadis,r.p.hoeksma\}@utwente.nl}}, Ruben Hoeksma\footnotemark[1],\\ S\'{a}ndor Kisfaludi-Bak\thanks{Institute for Theoretical Studies, ETH Z\"urich, Zurich, Switzerland. \texttt{sandor.kisfaludi-bak@mpi-inf.mpg.de}}, and Kevin Schewior\thanks{Department of Mathematics and Computer Science, University of Cologne, Cologne, Germany. \texttt{schewior@cs.uni-koeln.de}}}

\date{September 2021}

\begin{document}

\maketitle

\begin{abstract}
    We consider the online search problem in which a server starting at the origin of a $d$-dimensional Euclidean space has to find an arbitrary hyperplane. The best-possible competitive ratio and the length of the shortest curve from which each point on the $d$-dimensional unit sphere can be seen are within a constant factor of each other. We show that this length is in $\Omega(d)\cap O(d^{3/2})$.
\end{abstract}

\section{Introduction}

In $d$-dimensional Euclidean space, given a set of objects $\mathcal{T}\subseteq\mathcal{P}(\Reals^d)$, 
the \emph{online search problem (for an object in $\mathcal{T}$)} asks for 
a curve (or \emph{search strategy}) $\zeta:\Reals_{\geq 0}\rightarrow 
\Reals^d$ with $\zeta(0)=0$ that minimizes the competitive ratio~\cite{KarlinMRS88}. 
We say that $\zeta$ is \emph{$c$-competitive} if there exists an
$\alpha\in\Reals_{\geq 0}$ such that 
\[
    \len(\zeta|_{[0,t^\star]}) \leq c \cdot \dist(0,T)+\alpha
\]
for all $T\in\mathcal{T}$, where $t^\star=\inf\{t: \zeta(t)\in T\}$. 
The infimum over all such $c$ is the \emph{competitive ratio} of $\zeta$. Problems of this type date back to Beck~\cite{Beck64} and Bellman~\cite{Bellman63}.

The arguably most basic and well-known such online search problem from this class is the version in which $d=1$ and $\mathcal{T}=\mathbb{R}$, also known as the cow-path problem, where the best-possible competitive ratio is $9$~\cite{BeckNewman1970}, achieved by visiting the points $(-2)^0,(-2)^1,\dots$ sequentially. In fact, the class of all $9$-competitive strategies has been investigated more closely~\cite{0001DJ19}.
Moreover, the competitive ratio can be improved to about $4.591$ using randomization~\cite{BeckNewman1970}. The problem has been generalized to more than two paths starting at the origin~\cite{KaoRT96}, searching various types of objects in the plane or lattice~\cite{Baeza-YatesCR93}, and many more scenarios. For a (slightly outdated) survey, see the book by Gal~\cite{Gal1980}. 

Another very natural generalization is the case of general $d$ and $\mathcal{T}$ being equal to~$\mathcal{H}^{d-1}$, the set of all hyperplanes in $\Reals^d$. We call this version the \emph{$d$-Dimensional Hyperplane Search Problem}.
When $\mathcal{T}$ is a set of affine subspaces, this is arguably the most interesting case. Note that, if $\mathcal{T}$ contains all affine subspaces of dimension~$d'\le d-2$, any constant competitive ratio is ruled out, even when $d$ is fixed.
Surprisingly, in addition to the case of $d=1$, results only seem to be known for $d=2$ and for offline versions~\cite{DumitrescuT16,AntoniadisFHS20}, leaving a gap in the online-algorithms literature. For $d=2$, it is conjectured that a logarithmic spiral that achieves a competitive ratio of about $13.811$ is optimal~\cite{Baeza-YatesCR93,FinchZhu05}. The present paper addresses the aforementioned gap by providing the first asymptotic results for $d\rightarrow\infty$. 

For the remainder of the paper, we will consider the essentially equivalent but arguably cleaner \emph{Sphere Inspection Problem}: The goal is to find a $d$-dimensional minimum-length closed curve, $\gamma$, that \emph{inspects the unit sphere} in $\Reals^d$, i.e., $\gamma$ sees every point $p$ on the unit sphere $S^{d-1}$. Here, we say that an object $O\subseteq\mathbb{R}^d$ \emph{sees} a point $p'$ on the surface of the unit sphere  if there is a point~$p\in O$ such that the line segment $pp'$ intersects the unit \emph{ball} exclusively at $p'$. Note that the curve is not required to start in the origin any more. Such a minimum-length curve exists  for any dimension~\cite{GhomiWenk2021+}. While this problem is trivial for $d\in\{1,2\}$, it has only been shown recently that the best-possible length for $d=3$ is $4 \pi$~\cite{GhomiWenk2021+}. No results for higher dimensions are known. Such visibility problems have also been considered from an algorithmic point of view, e.g.,~\cite{CarlssonJN99,DrorELM03}.

While the connection between hyperplane search and sphere inspection seems to be folklore (e.g.,~\cite{GhomiWenk2021+}), we state it formally and provide a short proof for completeness. Let $I$ be an interval, and for $\beta\in\Reals$, a set $T\subseteq\Reals^d$, and curve $\gamma: I\rightarrow\Reals^d$, we denote by $\beta\cdot T$ the set $\{\beta\cdot x: x\in T\}$ and by $\beta\cdot\gamma$ the curve~$I\rightarrow\Reals^d,t\mapsto \beta\cdot\gamma(t)$. 

\begin{proposition}\label{prop:equivalence}
    Let $f:\mathbb{N}\rightarrow\mathbb{R}_{\geq 0}$. The following statements are equivalent:
    \begin{itemize}
        \item[(i)]
        There exists a length-$O(f(d))$ closed curve in~$\mathbb{R}^d$ that inspects $S^{d-1}$.
        \item[(ii)]
        There exists an $O(f(d))$-competitive strategy for the $d$-Dimensional Hyperplane Search Problem.
    \end{itemize}
\end{proposition}
\begin{proof}
    We first show that (i) implies (ii). Fix $d\in\mathbb{N}$, let $\gamma$ be the length-$\ell$ curve for the Sphere Inspection Problem in $\mathbb{R}^d$, and let $\gamma_0$ be a point on $\gamma$. We show that there exists a $(12\ell)$-competitive strategy for the $d$-Dimensional Hyperplane Search Problem where the additive constant~$\alpha$ is equal to~$3\ell$. Our strategy works in an infinite number of phases, starting with Phase~$0$. In each Phase $i$, the strategy $\zeta$ consists of a straight line from the origin to $2^i\gamma_0$, the curve $2^i\cdot \gamma$, and a straight line back from $2^i\gamma_0$ to the origin. 
    
    In the following, we show that, in Phase $i$, all hyperplanes of distance at most $2^i$ from the origin are visited, and the total distance traversed in Phase $i$ is at most $3\cdot 2^i\cdot \ell$, which implies the claim. We show two parts separately:
    
    \begin{itemize}
        \item We show that all claimed hyperplanes are visited by $\zeta$. Consider Phase $i$ and a hyperplane $H\in\mathcal{H}^{d-1}$ at distance at most $2^i$ from the origin. Let $n$ be a unit normal vector of $H$. Note that $2^i\cdot\gamma$ sees both the point $2^i n$ and the point~$-2^i n$. Let $p_1$ be a point from which $2^i\cdot\gamma$ sees $2^i n$, and let $p_2$ be a point from which it sees $-2^i n$. Note that, if $\dist(0,H)=2^i$, then $p_1$ or $p_2$ may be on $H$. In that case we are done. Otherwise $p_1$ and $p_2$ are on different sides of $H$, and, by the intermediate value theorem, $2^i\cdot \gamma$ intersects $H$.
        
        \item We next bound the distance traversed by $\zeta$ in Phase $i$. Note that it suffices to show that the start point of $2^i\cdot\gamma$, which is equal to its end point, has distance at most $2^i \ell$ from the origin. By scaling, we can restrict to showing $\lVert \gamma_0\rVert_2\leq \ell$. Let $H$ be the hyperplane $\langle\gamma_0,x\rangle=0$. Note that $\gamma$ must intersect $H$ since otherwise $\gamma$ does not see $-\gamma_0/\lVert\gamma_0\rVert_2$, thus $\dist(\gamma_0,H)\le \ell$. Since this distance (of $\gamma_0$ to $H$) is precisely~$\lVert\gamma_0\rVert_2$, it follows that $\lVert \gamma_0\rVert_2\leq \ell$.
    \end{itemize}
    
\noindent We now show that (ii) implies (i). Again fix $d\in\mathbb{N}$. Assume that there exists a curve $\gamma$ starting in the origin and $\alpha\in \Reals_{\geq 0}$ such that, for all hyperplanes $H\in\mathcal{H}^{d-1}$, $\gamma$ visits $H$ after traversing at most a length of $c\cdot\dist(0,H)+\alpha$. We show that, for any $\varepsilon>0$, there exists a, not necessarily closed, length-$(c+\varepsilon)$ curve that inspects the sphere. This then implies the existence of a length-$(2c+\varepsilon)$ closed curve that inspects the sphere. We define 
\[
    \zeta:=\frac{\varepsilon}{\alpha}\cdot\gamma|_{[0,t^\star]}\text{, where }t^\star:=\sup_{H\in\mathcal{H}^{d-1} : \dist(0,H)\leq\nicefrac{\alpha}{\varepsilon}}\min\{t:\gamma(t)\in H\}\,.
\]
Note that the minimum exists because $H$ is closed. Since for each $H\in\mathcal{H}^{d-1}$ with $\dist(0,H)\leq\nicefrac{\alpha}{\varepsilon}$, $\gamma$ visits $H$ after traversing at most a length of $c\cdot\nicefrac{\alpha}{\varepsilon}+\alpha$, the length of $\gamma|_{[0,t^\star]}$ is at most $c\cdot\nicefrac{\alpha}{\varepsilon}+\alpha$, so the length of $\zeta$ is at most $c+\varepsilon$.

It remains to be shown that $\zeta$ indeed sees every point $p\in S^{d-1}$. Let $H$ be the hyperplane tangent to the unit sphere in $p$. Further let $t:=\min\{t: \gamma(t)\in\nicefrac{\alpha}{\varepsilon}\cdot H\}$. By definition of $t^\star$, we have $t\leq t^\star$. Further, by definition of $\zeta$, we have $\zeta(t)\in H$, implying that $\zeta(t)$ sees $p$.
\end{proof}

In Section~\ref{sec:prelim}, we give an auxiliary lemma. In Sections~\ref{sec:lower} and~\ref{sec:upper} we will use that lemma and show the following two theorems.

\begin{theorem}\label{thm:lower}
Any curve in $\Reals^d$ that inspects $S^{d-1}$ has length at least~$2d$.
\end{theorem}

\begin{theorem}\label{thm:upper}
There exists a closed curve $\gamma$ in $\Reals^d$ of length $(2d)^{3/2}$ that inspects $S^{d-1}$.
\end{theorem}

\section{An Auxiliary Lemma}
\label{sec:prelim}

The following lemma simplifies thinking about the Sphere Inspection Problem.

\begin{lemma}\label{lem:aux}
    Let $P\subset\Reals^d$ be the convex hull of some point set $V$. Then, the following two statements are equivalent.
    \begin{enumerate}[label=(\roman*)]
        \item We have $S^{d-1}\subset P$.\label{item:SsubsetP}
        \item The set $V$ sees every point $p\in S^{d-1}$.\label{item:Vseesp}
    \end{enumerate}
\end{lemma}
\begin{proof}

We start by showing that (i) implies (ii). Let $H$ be the hyperplane that is tangent to the unit sphere in $p$. If $H$ contains a point $v\in V$, then the lemma directly follows since the line segment $vp$ lies completely within $H$ and therefore does not intersect the unit sphere other than at $p$. So assume that there is no point $v\in V$ that is contained in $H$. Then, and since $P$ contains one point of the hyperplane ($p\in S^{d-1}$, which is contained in $P$ by \ref{item:SsubsetP}), there must exist two points $v_1,v_2\in V$ which are separated by~$H$. Without loss of generality, let $v_1$ be the point in the halfspace defined by $H$ whose interior is disjoint from $S^{d-1}$. Then the line segment $v_1p$ is completely contained in that halfspace, and since the interior of the halfspace is disjoint from $S^{d-1}$, $v_1$ can see $p$.

Now we show that (ii) implies (i). Towards a contradiction, assume that there exists $p^\star\in S^{d-1}$ with $p^\star\notin P$. Then consider a hyperplane $H$ that separates $p^\star$ from $P$, and let $n$ be a unit normal vector of $H$ pointing away from $P$. Consider the hyperplane $H'$ that is tangent to $S^{d-1}$ in $n$. Clearly, both $P$ and $S^{d-1}\setminus\{n\}$ are contained in one open halfspace defined by $H'$. Note that no point in this halfspace can see $n$. Therefore, no $v\in V$ can see $n$; a contradiction.
\end{proof}

\section{Lower Bound}
\label{sec:lower}

The goal of this section is to prove Theorem~\ref{thm:lower}. Towards this, let $\gamma$ be a curve in $\Reals^d$ that inspects the unit sphere. We cut $\gamma$ into a minimum number of contiguous portions of length at most $\delta$ for some fixed $\delta<2$. Let $\xi_1,\dots,\xi_n$ be the resulting tour portions, where $n=\lceil |\gamma|/\delta \rceil$. Choose a portion $\xi$, and let $x$ be its midpoint. Clearly $\xi$ is contained in the ball $B$ that has center $x$ and radius $\delta/2$. Further define $C$ to be the cone that is the intersection of all halfspaces that contain both $B$ and $S^{d-1}$ and whose defining hyperplanes are tangent to both $B$ and $S^{d-1}$. Note that the set of points on the sphere that can be seen by the curve $\xi$ can also be seen from the apex of $C$, as visualized by Figure~\ref{fig:cone}. This holds since the radius of $B$ is $\delta/2<1$. Note that a single point $p\in \Reals^d$ can see some subset of an open hemisphere $H$ of the unit sphere. Let $H_1,\dots,H_n$ denote a set of open hemispheres such that $H_i$ covers the portion of the sphere seen by $\xi_i$. Since $\gamma$ inspects the sphere, we have that $(H_i)_{i=1}^n$ covers the sphere. 

\begin{figure}
    \centering
    \begin{tikzpicture}
    \node[anchor=south west,inner sep=0] (image) at (0,0) {\includegraphics[width=\textwidth]{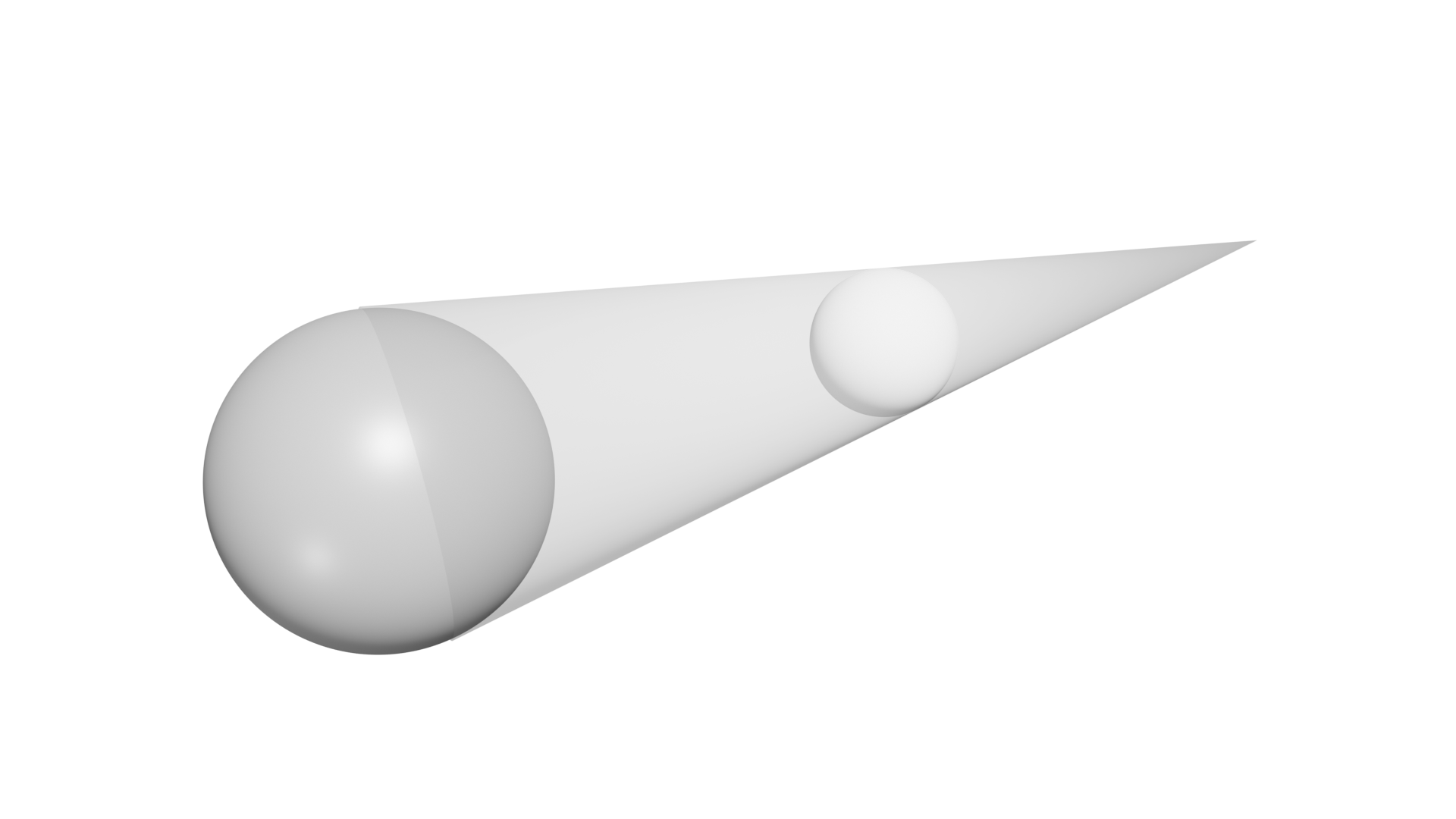}};
    \begin{scope}[x={(image.south east)},y={(image.north west)}]
        \node at (.6065,.581) {$B$};
        
        \node at (0.22,0.4) {$S^{d-1}$};
        
        \node at (0.45,0.62) {$C$};
        
        \draw (.607,.5816) circle (23.08pt);
        \coordinate (C) at (.26025,.413);
        \draw[rotate around={20:(C)}] (C) ellipse (54.7pt and 53.9pt);
        \draw (.31,.2183) -- (.862,.706) -- (.247,.625);
        \draw (.31,.22) to [bend right=11] (.247,.624);

        \draw[dashed] (.31,.22) to [bend left=11] (.247,.624);
        
        \draw[dotted] (.595,.8125) -- (0.61,.8);
        \draw[] (.61,.8) .. controls (.67,.75) .. (0.64,.63);
        \draw[dashed] (0.64,.63) .. controls (.61,.52) .. (.58,.57);        
        \draw[dashed,color=black!30] (.58,.57) .. controls (.57,.58) .. (.49,.5);
        \draw[color=black!30,dotted] (.487,.497) -- (.477,.487);
        \node at (.6,.515) {\scriptsize $\xi$};
    \end{scope}
\end{tikzpicture}
    \caption{The apex of the cone $C$ sees all points on the unit sphere that $B$ does.}
    \label{fig:cone}
\end{figure}
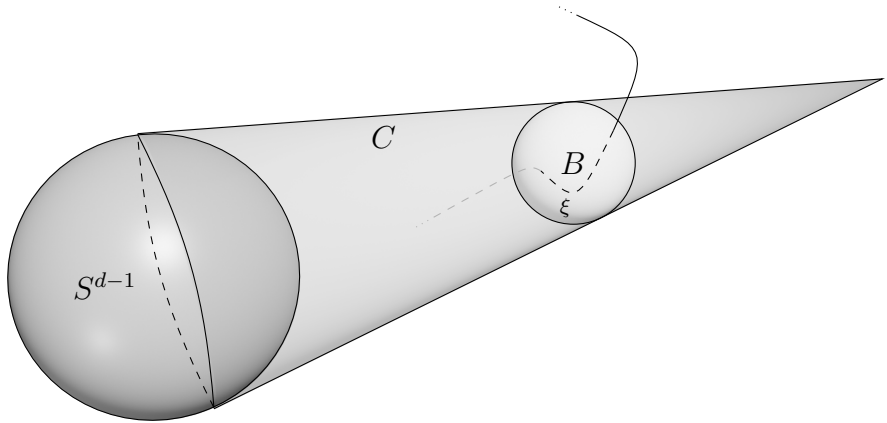

We need the following lemma.

\begin{lemma}\label{lem:hemicover}
The minimum number of open hemispheres that cover $S^{d-1}$ is $d+1$.
\end{lemma}

\begin{proof}
To show that $d+1$ points are sufficient, choose a simplex containing the sphere. Now consider the set of open hemispheres whose poles are colinear with the origin and a vertex of that simplex. Indeed, since the simplex contains the sphere, by Lemma~\ref{lem:aux}, the set of vertices of the simplex sees every point $p\in S^{d-1}$. Since each point  sees only a subset of the corresponding open hemisphere, the upper bound follows.

For the lower bound, we can use induction on $d$. Clearly the circle $S^1$ needs at least three open half-circles to be covered. For $S^{d-1}$, we have that the boundary of the first hemisphere $H_1$ is $S^{d-2}$, and each hemisphere $H_i$ can cover at most an open hemisphere of $\bd H_1$. So by induction at least $(d-2)+2=d$ hemispheres are needed to cover $\bd H_1$, and thus we have that at least $d+1$ hemispheres are needed to cover $S^{d-1}$.
\end{proof}

By Lemma~\ref{lem:hemicover} we have that $n\geq d+1$, implying $\lceil \len(\gamma)/\delta \rceil\geq d+1$ and therefore $\len(\gamma)/\delta \geq d$. With $\delta =2-\eps$, we have that $\len(\gamma) \geq (2-\eps)\cdot d$ for all $\eps>0$, and thus $\len(\gamma) \geq 2d$.

\section{Upper Bound}
\label{sec:upper}

In this section we prove Theorem~\ref{thm:upper}.
Let $C^d$ be the $d$-dimensional cross-polytope, i.e., the polytope $\{x\in\Reals^d: \lVert x\rVert_1 \le 1\}$. Define $\bar{C}^d:=\sqrt{d}\cdot C^d= \{x\in\Reals^d: \Vert x\rVert_1 \le \sqrt{d}\}$, as shown in Figure~\ref{fig:cross_polytope}. We claim that the vertices of $\bar{C}^d$ inspect the unit sphere.

\begin{lemma}
\label{lem:visibility}
For any point $p\in S^{d-1}$, there exists a vertex $v$ of $\bar{C}^d$ that sees $p$.
\end{lemma}
\begin{proof}
We prove that $\bar{C}^d$ contains the unit sphere; then Lemma~\ref{lem:aux} shows the claim. In other words we prove that for any point $x\in\Reals^d$ with $\lVert x\rVert_2 \le 1$ (so, $x\in S^{d-1}$), it also holds that $\lVert x\rVert_1 \le \sqrt{d}$. Indeed, by the Cauchy-Schwartz inequality we have that $\lVert x\rVert_1\le \sqrt{d}\cdot \lVert x\rVert_2$ for any $x$ which in turn can be upper bounded by $\sqrt{d}$ for any $x\in\Reals^d$ with $\lVert x\rVert_2\le 1$.
\end{proof}

Let $G(\bar{C}^d)$ be the graph\footnote{The graph $G(P)$ of a polytope $P$ is a graph with a node for each vertex of $P$ and an edge connecting those nodes if $P$ has an edge between the corresponding vertices. The graph of $P$ is also referred to as its $1$-skeleton.} of polytope $\bar{C}^d$. Note that $G(\bar{C}^d)$ is the so-called \emph{cocktail party graph} which can be obtained by removing a perfect matching from a complete graph on $2d$ vertices (indeed, in $\bar{C}^d$ any vertex $v$ has an edge to any other vertex except~$-v$). We next prove that $G(\bar{C}^d)$ is Hamiltonian.

\begin{lemma}
\label{lem:hamiltonian} 
The graph $G(\bar{C}^d)$ is Hamiltonian.
\end{lemma}
\begin{proof}
The proof is by induction on $d$. The statement clearly holds for $\Reals^2$ where the cross polytope is a square, and $G(\bar{C}^2)$ itself is a Hamiltonian cycle that we denote by $c^2$. Consider a Hamiltonian cycle $c^{d-1}$ for $G(\bar{C}^{d-1})$. Note that $G(\bar{C}^d)$ can be constructed by $G(\bar{C}^{d-1})$ and adding two nodes, $v^\star$ and $-v^\star$ that are connected to each of the nodes of $G(\bar{C}^{d-1})$. To construct a cycle $c^d$ of $G(\bar{C}^d)$, take any two distinct edges $\{v_1,v_2\}$ and $\{v_3,v_4\}$ contained in $c^{d-1}$ and replace them with the edges $\{v_1,v^\star\},\{v^\star,v_2\}$ and $\{v_3,-v^\star\},\{v_3,-v^\star\}$, respectively. The resulting tour is connected, visits all nodes of $G(\bar{C}^d)$ exactly once and the used edges are contained in the edge set of $G(\bar{C}^d)$, thus ensuring the feasibility of~$c^d$.
\end{proof}

\begin{figure}
    \centering
    \begin{tikzpicture}
        \node[anchor=south west,inner sep=0] (image) at (0,0) {\includegraphics[width=\textwidth]{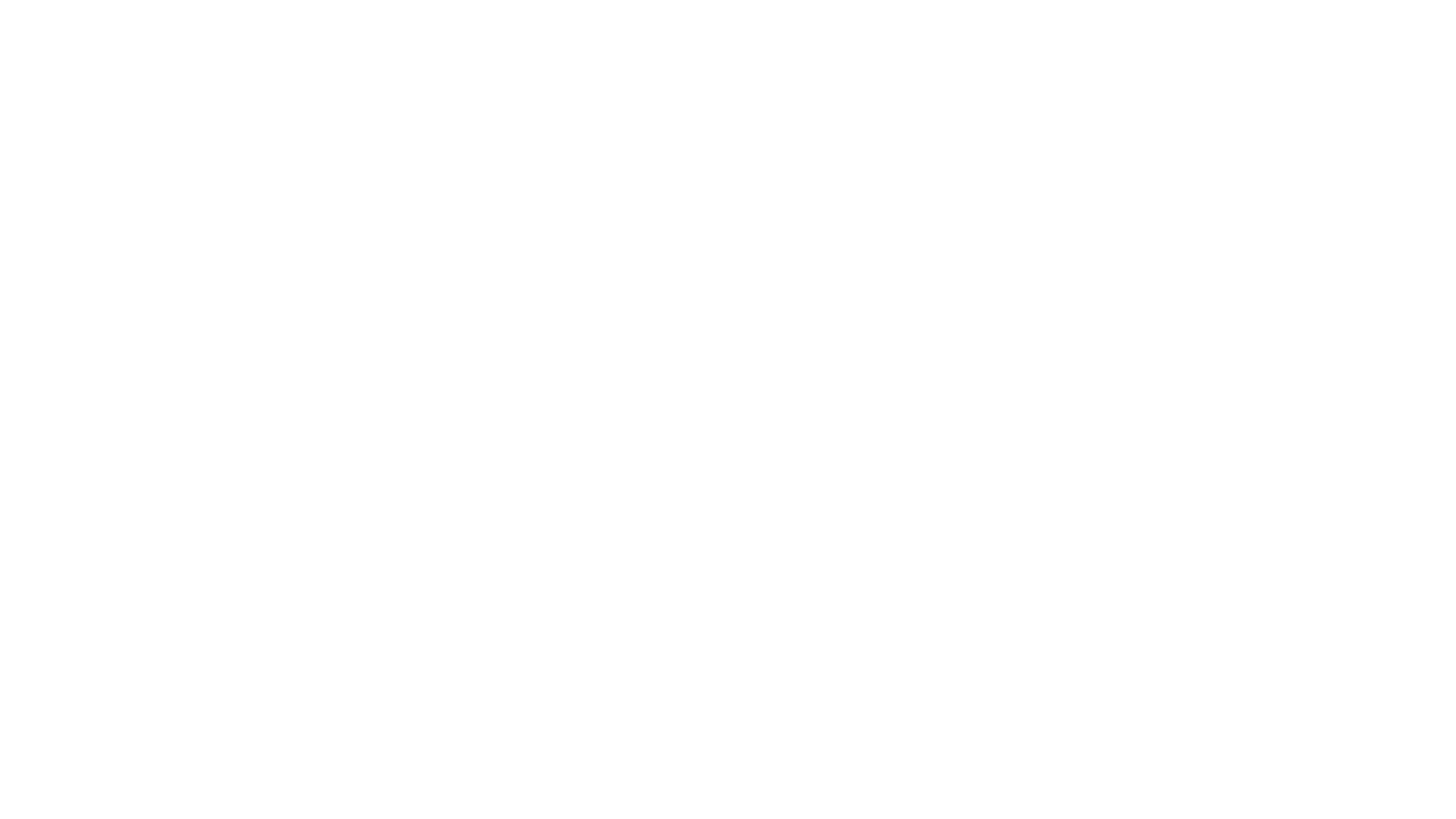}};
        \begin{scope}[x={(image.south east)},y={(image.north west)}]
            \path (.285,0.513) coordinate (XN) -- (.4825,.9) coordinate (ZP) -- (0.707,.462) coordinate (XP)-- (.484,.103) coordinate (ZN) -- cycle;
            \coordinate (YP) at (0.3825,.415);
            \coordinate (YN) at ($(XN)+(XP)-(YP)$); 
        \end{scope}
        \draw[dashed] (ZP) -- (YN);
        \draw[dashed] (XP) -- (YN);
        \draw[very thick, dashed] (XN) -- (YN) -- (ZN);
        \node[anchor=south west,inner sep=0] (image) at (0,0) {\includegraphics[width=\textwidth]{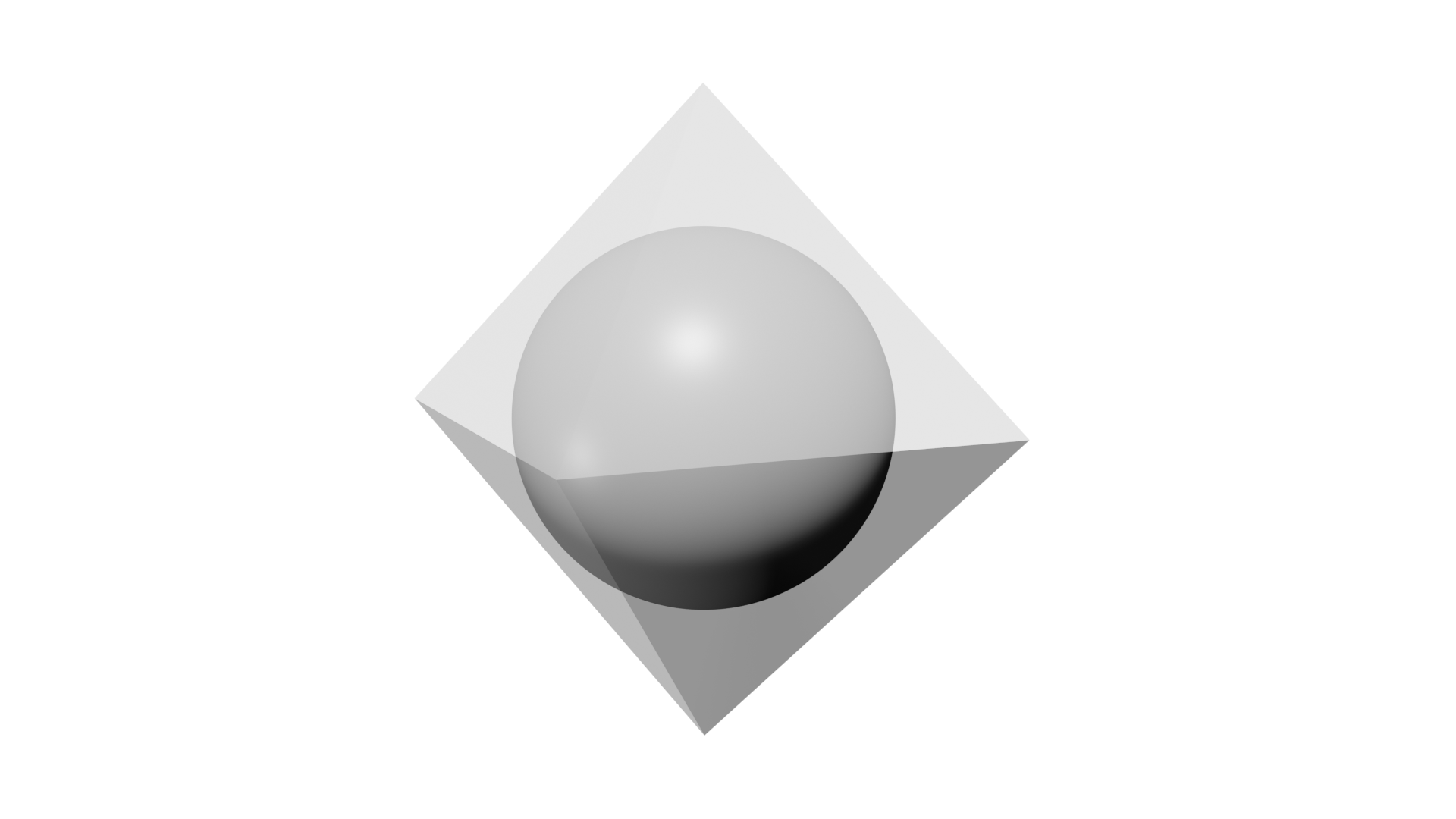}};
        \draw (ZP) -- (YP) -- (ZN);
        \draw (XP) -- (YP) -- (XN);
        \draw (XN) -- (ZP) -- (XP) -- (ZN) -- cycle;
        \draw[very thick] (ZN) -- (YP) -- (XP) -- (ZP) -- (XN);
    \end{tikzpicture}
    \caption{The cross polytope that has vertices at distance $\sqrt{d}$ from the origin contains the unit sphere, and thus, a tour of its vertices (thick) inspects the unit sphere.}
    \label{fig:cross_polytope}
\end{figure}
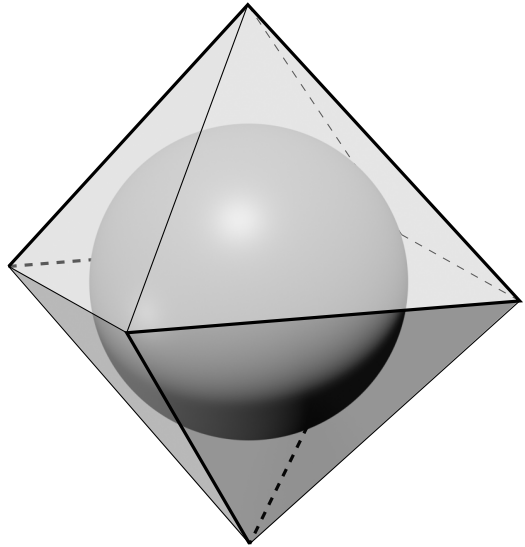

By Lemma~\ref{lem:visibility} it directly follows that any closed curve that visits all the vertices of $\bar{C}^d$ inspects the unit sphere, and therefore so does the closed curve $\gamma$ corresponding to the Hamiltonian cycle ${c}^d$ in the skeleton of $\bar{C}_d$. To complete the proof of Theorem~\ref{thm:upper},
note that each edge of $\bar{C}_d$ has a length of $\sqrt{2d}$ and that $\gamma$ traverses $2d$ such edges.

\section{Conclusion}
In this paper, we narrowed down the optimal competitive ratio for the $d$-Dimensional Hyperplane Search Problem to $\Omega(d)\cap O(d^{3/2})$. The obvious open problem is closing this gap.

\paragraph*{Acknowledgements.} We thank Paula Roth for helpful discussions.

\bibliographystyle{plain}
\bibliography{searching}

\end{document}